\documentclass[runningheads]{llncs}
\usepackage{graphicx,color,amsfonts,amsmath}
\usepackage[hidelinks]{hyperref}

\usepackage{tikz}

\newcommand{\bC}{\mathbb{C}}
\newcommand{\bR}{\mathbb{R}}
\newcommand{\bQ}{\mathbb{Q}}
\newcommand{\bN}{\mathbb{N}}
\newcommand{\real}{{\rm real}}
\newcommand{\imag}{{\rm imag}}
\newcommand{\maple}{{\tt Maple}}
\newcommand{\codehome}{\url{https://github.com/P-Edwards/EvalCertification}}

\spnewtheorem{algorithm}[theorem]{Algorithm}{\bfseries}{}

\begin{document}
\title{Certified evaluations of H\"older continuous functions
at roots of polynomials}
\titlerunning{Certified evaluations}

\author{Parker B. Edwards\inst{1}\orcidID{0000-0001-6875-5328} \and
\mbox{Jonathan D. Hauenstein}\inst{1}\orcidID{0000-0002-9252-8210} 
\and 
\mbox{Clifford D. Smyth}\inst{2}\orcidID{0000-0003-1486-7900}}
\authorrunning{P.B. Edwards et al.}

\institute{Department of Applied and Computational Mathematics and Statistics, University of Notre Dame, Notre Dame, IN 46556 \email{\{parker.edwards,hauenstein\}@nd.edu} \\
\url{https://sites.nd.edu/parker-edwards}, \url{https://www.nd.edu/\~jhauenst} \and 
Department of Mathematics and Statistics, University of North Carolina at Greensboro, 
Greensboro, NC 27402
 \email{cdsmyth@uncg.edu}\\
 \url{https://www.uncg.edu/~cdsmyth/}}

\maketitle              
\begin{abstract}
Various methods can obtain certified estimates for roots of polynomials. Many applications in science and engineering additionally utilize the value of functions evaluated at roots.  For example, critical values are obtained by
evaluating an objective function at critical points.
For analytic evaluation functions, Newton's method naturally applies to yield certified estimates.
These estimates no longer apply, however,
for H\"older continuous functions, 
which are a generalization of Lipschitz continuous functions
where continuous derivatives need not exist.
This work develops and analyzes 
an alternative approach
for certified estimates
of evaluating locally H\"older continuous functions
at roots of polynomials.~An implementation of the method in {\tt Maple}~demonstrates~efficacy~and~efficiency.  

\end{abstract}

\section{Introduction}\label{sec:Intro}

For a univariate polynomial $p(x)$, the Abel-Ruffini theorem posits 
that the roots cannot be expressed in terms of radicals for general polynomials of degree 
at least~$5$.  A simple illustration of this is that the solutions of the quintic equation
\begin{equation}\label{eq:SimpleQuintic}
p(x) = x^5 - x - 1 = 0
\end{equation}
cannot be expressed in radicals.
Thus, a common technique is to compute numerical approximations with certified bounds 
for the roots of a polynomial.  Some approaches based on Newton's method
are the Kantorovich theorem \cite{Kantorovich} 
and Smale's $\alpha$-theory \cite{SmaleOnePoint}.
Kantorovich's approach is based on bounds for a twice-differentiable function
in an open set while Smale's approach only uses local estimates
at one point coupled with the analyticity of the function.  
Certified approximations of roots of polynomials
can also be obtained using interval methods such as \cite{CircularRoots,IntroInterval,Rump}
along with the Krawczyk operator \cite{Krawczyk,Moore1977}.  

Although computing certified estimates for roots of polynomials
is important, many applications in science and engineering utilize
the roots in further computations.  As an illustrative example, 
consider the optimization problem
\begin{equation}\label{eq:SimpleOpt}
\min \{21x^8 - 42x^4 - 56x^3 + 3~:~x\in\bR\}.
\end{equation}
The global minimum is the minimum of the critical values
which are obtained by evaluating
the objective function $g(x) = 21x^8 - 42x^4 - 56x^3 + 3$
at its critical points, i.e. at the real roots of $g'(x) = 168 x^2 (x^5 - x - 1)$.
Since the quintic in~\eqref{eq:SimpleQuintic} is a factor of $g'(x)$, 
only approximations of the roots of $g'$ can be computed. One must translate these approximate roots to certified evaluations of the objective function $g(x)$ evaluated at the roots of the polynomial~$g'(x)$ to obtain certified bounds on the global minimum of
\eqref{eq:SimpleOpt}.

One approach for computing a certified evaluation of $f(x)$
at roots of a polynomial $p(x)$ is via certified estimates
of solutions to the multivariate system
\begin{equation}\label{eq:MultivariateEval}
\left[
\begin{array}{c}
p(x) \\
y - f(x) \end{array}\right] = 0.
\end{equation}
For sufficiently smooth $f$,
approaches based on Newton's method
generate certified estimates.
The downside is that this requires $f$ to be differentiable.
Alternatively, one can follow a two-stage 
procedure: develop certified
bounds of a root of~$p(x)$ and then use interval
evaluation methods, e.g., see \cite[Chap.~5]{IntroInterval},
to develop certified bounds on $f(x)$
evaluated at the root. This two-stage approach does not allow direct control on the precision of the certified~evaluation. 

The approach in this paper considers certified evaluations of locally H\"older continuous
functions at roots of polynomials and links the desired output
of the certified evaluation with the error in the 
approximation of the root.  H\"older continuous functions are a generalization of Lipschitz functions
which are indeed continuous, but they need not be 
differentiable anywhere. E.g. Weierstrass function is locally H\"older.
Satisfying the local H\"older continuity condition does not guarantee that a function can be evaluated
exactly for, say, rational input. Our approach also incorporates
numerical evaluation error into the certified bounds to address this issue.

The rest of the paper is organized as follows.
Section~\ref{sec:Holder} describes the necessary analysis 
of locally H\"older continuous functions, with a particular 
focus on polynomials and rational functions.
Section~\ref{sec:Roots} summarizes the approach used for developing
certified bounds on roots of polynomials.  
Section~\ref{sec:CertEval} combines
the certification of roots and evaluation bounds
on H\"older continuous functions yielding
our approach for computing certified evaluations.
Section~\ref{sec:Examples} presents
information regarding the implementation
in {\tt Maple} along with several examples
demonstrating its efficacy and efficiency.
Section~\ref{sec:Application} applies 
the techniques developed for certified evaluations
to prove non-negativity of coefficients arising
in a series expansion of a rational function.
The paper concludes in Section~\ref{sec:Conclusion}.

\section{H\"older continuous functions}\label{sec:Holder}

The following describes the collection of functions under consideration.

\begin{definition}
A function $f:\bC\rightarrow\bC$ is {\em locally H\"older continuous} at a point \mbox{$x^*\in\bC$}
if there exist positive real constants $\epsilon,C,\alpha$ such that
\begin{equation}\label{eq:HolderBound}
|f(x^*)-f(y)|\leq C\cdot|x^*-y|^\alpha \leq C\cdot \epsilon^\alpha
\end{equation}
for all $y\in B(x^*,\epsilon)$ where $B(x^*,\epsilon) = \{z\in\bC~:~|z-x^*|\leq \epsilon\}$.
In this case, $f(x)$ is said to have {\em H\"older constant} $C$
and {\em H\"older exponent}~$\alpha$ at $x^*$.
Moreover, if $\alpha=1$, then 
$f(x)$ is said to be {\em Lipschitz continuous}
at $x^*$ with {\em Lipschitz constant}~$C$.
\end{definition}

Functions which are locally H\"older continuous 
at a point are clearly continuous at that point and
the error bound provided in \eqref{eq:HolderBound} will be exploited in 
Section~\ref{sec:CertEval} to provide
certified evaluations.
Every function $f(x)$ which is continuously 
differentiable in a neighborhood of $x^*$ is locally H\"older continuous
with $\alpha = 1$, i.e., locally Lipschitz continuous.
For $n \geq 1$, $f(x) = \sqrt[n]{|x|}$ is continuous but not differentiable at $x^*=0$.
It is locally H\"older continuous at $x^*=0$ with H\"older constant $C = 1$ and H\"older exponent~\mbox{$\alpha = 1/n$}.

A computational challenge is to determine a H\"older constant $C$ and H\"older exponent $\alpha$ for $f(x)$ on $B(x^*,\epsilon)$ given 
$f(x)$, $x^*$, and $\epsilon > 0$.  
Sections~\ref{sec:HolderPoly} and~\ref{sec:HolderRat}
describe a strategy for polynomials and 
rational functions,~respectively. 

\subsection{Polynomials}\label{sec:HolderPoly}

Since every polynomial $f(x)$ is continuously differentiable, 
we can take the H\"older exponent to be $\alpha = 1$ at any point $x^*$.  
However, the H\"older constant $C$ depends upon $x^*$ and $\epsilon$.
The Fundamental Theorem of Calculus
shows that one just needs
\begin{equation}\label{eq:HolderDerivative}
C \geq \max_{y\in B(x^*,\epsilon)} |f'(y)|.
\end{equation}
Although one may attempt to compute this maximum, 
Taylor series expansion of $f'(x)$ at $x^*$ provides an easy to compute upper bound.  If $d = \deg f$, 
$$f'(x) = \sum_{i=1}^{d} \frac{f^{(i)}(x^*)}{(i-1)!} (x-x^*)^{i-1},$$
so that the triangle inequality yields
\begin{equation}\label{eq:ComputeMax}
C := \sum_{i=1}^d \frac{|f^{(i)}(x^*)|}{(i-1)!}\epsilon^{i-1}
\geq \max_{y\in B(x^*,\epsilon)} |f'(y)|.
\end{equation}

\subsection{Rational functions}\label{sec:HolderRat}

The added challenge with a rational function $f(x)$ is to ensure that it
is defined on $B(x^*,\epsilon)$.  One may attempt to compute the poles of $f(x)$
and ensure that none are in $B(x^*,\epsilon)$, the implementation in Section \ref{sec:Examples} 
is based on the following local approach
that also enables computing local upper bounds on $|f'(x)|$.
For $f(x) = a(x)/b(x)$, one can prove $b(y) \neq 0$ for all $y\in B(x^*,\epsilon)$
by showing that \mbox{$|b(x^*)| > |b(y)-b(x^*)|$} for all $y\in B(x^*,\epsilon)$.  
If $d_b = \deg b$, then
$$|b(y)-b(x^*)| = \left|\sum_{i=1}^{d_b} \frac{b^{(i)}(x^*)}{i!}(y-x^*)^i\right|
\leq \sum_{i=1}^{d_b} \frac{|b^{(i)}(x^*)|}{i!}\epsilon^i.$$
Therefore, a certificate that $f(x)$ is continuously differentiable on $B(x^*,\epsilon)$ is
$$|b(x^*)| > \sum_{i=1}^{d_b} \frac{|b^{(i)}(x^*)|}{i!}\epsilon^i$$
in which case 
\begin{equation}\label{eq:Minimum}
\min_{y\in B(x^*,\epsilon)}|b(y)| \geq |b(x^*)| - \sum_{i=1}^{d_b} \frac{|b^{(i)}(x^*)|}{i!}\epsilon^i > 0.
\end{equation}
When $b(x^*)\neq 0$, it is clear
that one can always take $\epsilon$ small
enough to satisfy~\eqref{eq:Minimum}.

When $f(x)$ is continuously differentiable on $B(x^*,\epsilon)$, then
one can take the H\"older exponent $\alpha = 1$ and 
the H\"older constant $C$ as in \eqref{eq:HolderDerivative}.
Hence,
$$\max_{y\in B(x^*,\epsilon)} |f'(x)| \leq
\dfrac{\displaystyle\max_{y\in B(x^*,\epsilon)} |a'(y)|}{\displaystyle\min_{y\in B(x^*,\epsilon)} |b(y)|}
+ \dfrac{\displaystyle\max_{y\in B(x^*,\epsilon)} |a(y)|\cdot
\displaystyle\max_{y\in B(x^*,\epsilon)} |b'(y)|}{\displaystyle\min_{y\in B(x^*,\epsilon)} |b(y)|^2}$$
where the maxima can be upper bounded similar to \eqref{eq:ComputeMax}
and the minimum can be lower bounded using \eqref{eq:Minimum}.

\section{Certification of roots}\label{sec:Roots}

The initial task of determining certified evaluation bounds
at roots of a given polynomial is to compute certified bounds
of the roots.
From a theoretical perspective, we assume that we know the
polynomial $p(x)$ exactly.  From a computational perspective,
we assume that $p(x)$ has rational coefficients, i.e., $p(x)\in\bQ[x]$.
The certification of roots of $p(x)$ can thus be performed using
{\tt RealRootIsolate} based on \cite{RealRoots1,RealRoots2,RealRoots4,RealRoots5,RealRoots3} 
in {\tt Maple} as follows.  

Since $p(x)$ is known exactly, we first reduce down to the irreducible case
with multiplicity $1$ roots by computing an irreducible factorization of $p(x)$, say
$$p(x) = p_1(x)^{r_1}\cdots p_s(x)^{r_s}$$
where $p_1,\dots,p_s$ are irreducible with corresponding
multiplicities $r_1,\dots,r_s\in\bN$.  
For $p(x)\in\bQ[x]$, {\tt factor} in \maple
computes the irreducible factors in $\bQ[x]$,
i.e., each $p_i(x)\in\bQ[x]$.
If $z\in\bC$ is a root of $p_j(x)$, then $z$ has multiplicity $1$ with respect to~$p_j(x)$, i.e.,~$z$ is a simple root of $p_j(x)$ with 
$p_j(z)=0$ and $p_j'(z)\neq 0$. In contrast, $z$ has multiplicity~$r_j$ with respect to~$p(x)$.

For each irreducible factor $q := p_j$, we transform the domain $\bC$ into $\bR^2$~via
\begin{equation}\label{eq:RealImagQ}
q(x+iy) = q_r(x,y) + i\cdot q_i(x,y) \,\,\,\,\,\,\,\,\,\,\,\,\,\,\,\,\,\,\,\, 
\hbox{where $i=\sqrt{-1}$}.
\end{equation}
Therefore, solving $q=0$ on $\bC$
corresponds with solving \mbox{$q_r=q_i=0$} on~$\bR^2$.  
Applying {\tt RealRootIsolate} with an optional 
absolute error bound {\tt abserr} that will be utilized later
guarantees as output isolating boxes 
for every real solution to \mbox{$q_r=q_i=0$} on~$\bR^2$.
Therefore, looping over the irreducible factors of $p$, one obtains certified
bounds for every root $z$ of $p(x)$ in $\bC$
of the form $a_1 \leq \real(z) \leq a_2$ and $b_1 \leq \imag(z) \leq b_2$
where $a_1,a_2,b_1,b_2\in\bQ$.

One can enhance this certification to provide additional information about the
roots.  For example, suppose one aims to classify the roots of $p(x)$
in~$\bC$ based on their modulus being less than $1$, equal to $1$, and 
greater than $1$.  The {\tt RealRootIsolate} command can be applied
to \mbox{$q_r(x,y)=q_i(x,y)=0$} for 
\mbox{$(x,y)\in\bR^2$}
together with $x^2+y^2-1<0$, $x^2+y^2-1=0$,
and $x^2+y^2-1>0$, respectively, to perform this classification.

\section{Certified evaluations}\label{sec:CertEval}

Combining information on H\"older continuous functions
from Section~\ref{sec:Holder} and certification of roots
of polynomials from Section~\ref{sec:Roots} yields the following approach
to develop certified evaluations.  
With input a polynomial $p(x)$, a H\"older continuous function $f(x)$
which is defined at each root of $p(x)$, 
and an error bound $\epsilon>0$, the goal is to develop an approach
that computes $f(z)$ within $\epsilon$ for each root $z$ of $p(x)$.
Since one may not be able to evaluate $f(x)$ exactly, we incorporate
an evaluation error of $\delta\in (0,\epsilon)$.  Typically,
$\delta$ can be decreased by utilizing higher precision computations.  For rational input,
rounding to produce finite decimal representations 
constitutes the only source of representation error in {\tt Maple}.  Our implementation utilizes enough digits
to have $\delta=\frac{\epsilon}{10}$.

The first step is to utilize Section~\ref{sec:Roots} to determine initial
certified bounds for each root of $p(x)$.  For an initial error bound on the roots, 
we start with $\gamma = \epsilon/2$ so that each root $z$ is approximated by
$x^*$ with $|z-x^*|<\gamma$.  Certified evaluations are then obtained root by
root since the H\"older constants are dependent upon local information near each root.
In particular, the next step is to compute a H\"older exponent $\alpha$
and H\"older constant $C$ that is valid on the ball $B(x^*,2\gamma)$.  
If this is not possible, e.g., if $f(x)$ cannot be certified to be defined
on $B(x^*,2\gamma)$, one can simply reduce $\gamma$, e.g., by replacing $\gamma$ by $\gamma/2$, 
and repeat the process using a newly computed certified approximation of $z$.  
Since $f(x)$ is defined at each root of $p(x)$, this loop must terminate.

The final step is to utilize local information
to compute a new approximation of root $z$ that will produce a certified 
evaluation within $\epsilon$.  Consider $\mu$ such that
$$0 < \mu\leq \min\left\{\gamma,\sqrt[\alpha]{\frac{\epsilon-\delta}{C}}\right\}$$
and $z^*$ an approximation of $z$ such that $z\in B(z^*,\mu)$.  
Since $|x^*-z^*|\leq 2\gamma$, we have $B(z^*,\mu)\subset B(x^*,2\gamma)$
so that all of the H\"older constants are valid on $B(z^*,\mu)$.  Hence,
all that remains is to compute a certified approximation of $f(z^*)$, say $f^*$,
within the evaluation error of $\delta$ since
$$|f^*-f(z)|\leq |f^*-f(z^*)|+|f(z^*)-f(z)| \leq \delta + C\cdot |z^*-z|^\alpha \leq
\delta + C\cdot \mu^\alpha \leq \epsilon.$$

\section{Implementation and examples}\label{sec:Examples}

The certified evaluation procedure has been implemented as a \maple \ package
entitled {\tt EvalCertification} available at \codehome \ along with \maple \  notebooks for the examples. 
The main export is the procedure {\tt EstimateRootsAndCertifyEvaluations}
with the following high level signature. \\
{\bf Input:} 
\begin{itemize}
    \item Univariate polynomial $p\in\mathbb{Q}[x]$.
    \item List of locally H\"older continuous functions $f_1,\dots,f_m$ with which to certifiably estimate evaluations at the roots of $p(x)$.
    \item List of procedures specifying how to compute local H\"older constants and exponents for $f_1,\dots,f_m$. (See Section~\ref{subsec:holder_extension} for example of the syntax).
    \item Desired accuracy $\epsilon\in\bQ_{>0}$.
\end{itemize}
{\bf Main output:}
\begin{itemize}
    \item Complex rational root approximations $z_1^*,\dots,z_s^*$, one for each of the distinct roots $z_1,\dots,z_s$ of $p(x)$, such that $\vert z_j - z_j^* \vert \leq \epsilon$.
    \item For each $f_i$ and $x_j$, a complex decimal number $f_{ij}^*$ with $\vert f_i(x_j) - f_{ij}^*\vert \leq \epsilon$. 
\end{itemize}
The {\tt EvalCertification} package is formatted in a .mpl file which can be 
read into a notebook with: 

\begin{verbatim}
    read("EvalCertification.mpl")
    with(EvalCertification)
\end{verbatim}
This lists the package's following four exports: the main function and three built in procedures for determining local H\"older constants and exponents for common classes of H\"older functions.

{\color{blue} \it \small
\begin{verbatim}
  EstimateRootsAndCertifyEvaluations, HolderInformationForExponential,
  HolderInformationForPolynomial,  HolderInformationForRationalPolynomial
\end{verbatim}
}

The following examples highlight the specific \maple \ types of these inputs and 
outputs as well as other interface details. 

\subsection{Critical values}\label{subsec:CriticalValues}

As a first example, consider \eqref{eq:SimpleOpt} by
certifiably evaluating 
\[
f(x) = 21x^8-42x^4-56x^3+3
\] 
at the roots of $p(x) = f'(x) = 168x^2(x^5-x-1)$.
\begin{verbatim}
    f_polynomial := 21x^8 - 42x^4 - 56x^3 + 3;
    f_derivative := diff(f_polynomial, x);
    EstimationPrecision := 1/10^14
\end{verbatim}
The main call to {\tt EstimateRootsAndCertifyEvaluations} is subsequently:
\begin{verbatim}
    solutions_information := 
    EstimateRootsAndCertifyEvaluations(f_derivative, 
                            [f_polynomial, f_derivative], 
                            HolderInformationForPolynomial, 
                            EstimationPrecision);
\end{verbatim}

The first argument provides the polynomial to solve and the second a list of polynomials to evaluate. 
For illustration, we include evaluating the polynomial to solve in the evaluation list.  
The third argument is a procedure for computing H\"older constants which, in this case,
uses the procedure that implements the estimates in Section~\ref{sec:HolderPoly}
for polynomials.  Notice that we need only provide the procedure once since all functions for evaluation fall into the same class of H\"older functions, namely polynomials.
The last argument is the final error bound.

The output {\tt solutions\_information} is formatted as a Record. Certifiably estimated roots are stored in a list as illustrated.
{\color{blue}\begin{verbatim}
    solutions_information:-root_values =
    [0, 
    2691619717901426047/2305843009213693952, 
    26745188167908553113/147573952589676412928 - 
    19995423894655642147*I/18446744073709551616, ...]
\end{verbatim}
} 
Evaluations are also stored in lists, one list for each function to evaluate with one entry for each root of $p$. Estimates are ordered so that the estimate at index~$i$ in its list corresponds to the root at index~$i$ in the roots list.
{\color{blue}\begin{verbatim}
    solutions_information:-evaluations_functions_1 =
    [3., -91.6600084778015707, ...];
    solutions_information:-evaluations_functions_2 =
    [0, -6.692143197043304*10^(-16),...];
\end{verbatim}
}
Therefore, the solution to \eqref{eq:SimpleOpt} is $-91.6600084778015707$
which is certifiably correct within an error of $10^{-14}$.

\subsection{Extending with custom H\"older information procedures} \label{subsec:holder_extension}

Polynomial and rational functions can utilize
the built-in procedures for computing local
H\"older constants.  
One more feature of {\tt EvalCertification} 
is the ability to extend the certification procedures
to new classes of functions by specifying how to 
compute local H\"older constants.  
As an example, consider $f(x)=\vert x\vert^\frac{1}{r}$ where $r > 1$. As noted in Section \ref{sec:Holder}, 
$f$ is globally H\"older continuous with 
H\"older constant $C=1$ and H\"older exponent $\alpha = \frac{1}{r}$. 
The following shows the necessary format
for defining a new procedure to compute
the local H\"older information.
\begin{verbatim}
    alpha := 1/r;
    ExponentialInfo := proc(InputFunction,point,radius)
        return(Record("exponent"=alpha,"constant"=1,
                      "domain_estimate"=false));
    end proc;
\end{verbatim}
All custom H\"older information procedures must follow the same signature as this example. The exponent $\alpha$ and constant $C$ in the output Record should satisfy the H\"older conditions for {\tt InputFunction} on the ball $B(\texttt{point},{\texttt{radius}})\subset\bC$. The domain estimate lists points within {\tt radius} of points missing from the input function's domain, or false if defined everywhere.

To illustrate, consider $f_1(x) = |x|^{1/3}$
and
$$f_2(x) = \frac{-6x^5-8x^4+18x^3+88x^2-54}{6x^5+30x^4+16x^3-18x^2-44x}.$$
The following commands 
produce certified evaluations of $f_1(x)$ and $f_2(x)$ 
at the roots of $p(x) = 21x^8 - 42x^4 - 56x^3 + 3$
from Section~\ref{subsec:CriticalValues}
using the above procedure {\tt ExponentialInfo} for $r=3$.
\begin{verbatim}
    f_polynomial := 21x^8 - 42x^4 - 56x^3 + 3;
    f1 := abs(x)^(1/3);
    f2:=(-6x^5-8x^4+18x^3+88x^2-54)/(6x^5+30x^4+16x^3-18x^2-44x);
    EstimationPrecision := 1/10^14;
    solutions_information := 
    EstimateRootsAndCertifyEvaluations(
        f_polynomial,
        [f1,f2],
        [ExponentialInfo,HolderInformationForRationalPolynomial],
        EstimationPrecision);
\end{verbatim}
The third argument is a list instead of single procedure. This is necessary when evaluating functions that require different approaches for computing H\"older information. The built-in procedure {\tt HolderInformationForExponential} 
takes as input a single number $\alpha$ and outputs the procedure {\tt ExponentialInfo} for that $\alpha$. We could equivalently replace {\tt ExponentialInfo} with the procedure {\tt HolderInformationForExponential(1/3)} in this example.

\subsection{Benchmarking}\label{subsec:benchmarking}

The dominant computational cost is in estimating roots and computing local H\"older constants.
Suppose that $R(p,\epsilon)$ is
the complexity of approximating roots of $p$ within $\epsilon$, 
$H(f_1,\dots,f_n,p,\epsilon)$ is the minimum complexity of
computing H\"older constants at one root, and $A(p,f_1,\dots,f_n,\epsilon)$ is the number of repetitions required to find an accuracy $\gamma \leq \epsilon$ where local H\"older constants can be calculated. Then \[A(p,f_1,\dots,f_n,\epsilon)(R(p,\epsilon) + n\deg(p)H(f_1,\dots,f_n,p,\epsilon))+R(p,\epsilon)\] 
is a lower bound on the complexity.
The number of repetitions $A$ depends on the input in a complicated way which we do not attempt to characterize here. 

We benchmarked {\tt EvalCertification}
using random polynomials generated by 
the command {\tt randpoly} in \maple.
All tests computed roots of a random polynomial $p(x)$ 
with integer coefficients between $-10^{10}$ and $10^{10}$
and evaluated rational functions where the numerator
and denominator were polynomials of degree $D$.
The average was taken over 50 random selections.
Figure~\ref{fig:benchmark} shows the results of the benchmarking tests, which were performed
on Ubuntu 18.04 running \maple \ 2020 
with an Intel Core i7-8565U processor.
They were based on the degree~$d$ of $p(x)$, the value of $D$,
the number of functions $n$ to evaluate, 
and the size of the output error $\epsilon$.

\begin{figure}[!t]
    \centering
\resizebox{\textwidth}{!}{
    \begin{tabular}{cccc}
    \includegraphics[scale=1]{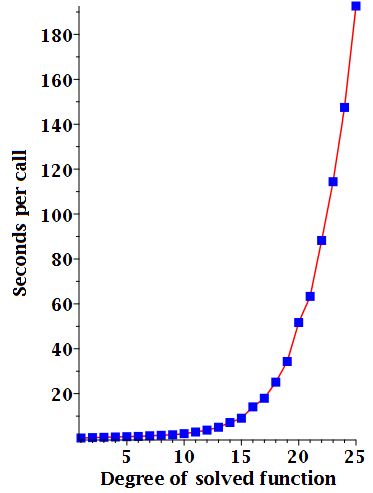} &
    \includegraphics[scale=1]{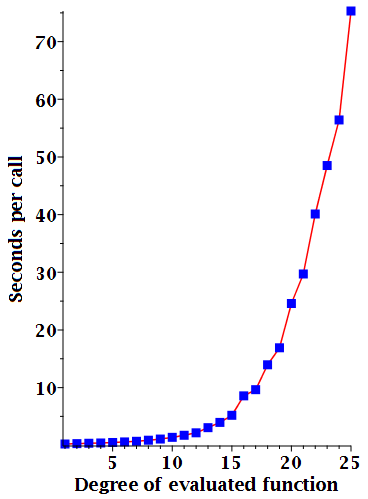} &
    \includegraphics[scale=1]{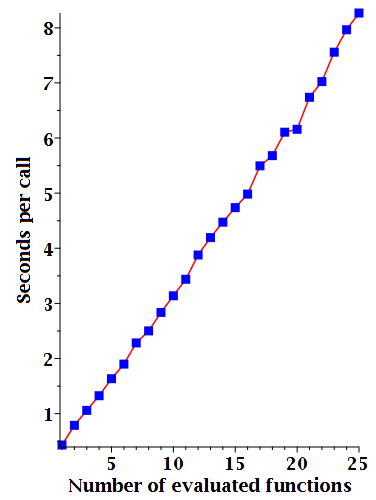} &
    \includegraphics[scale=1]{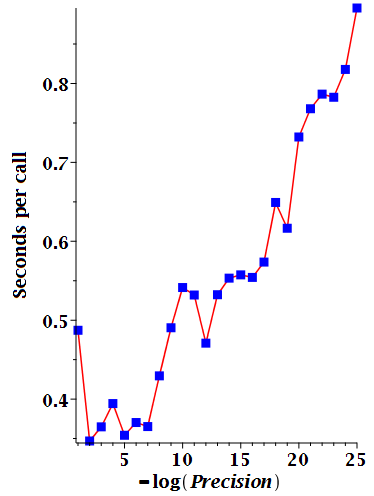}
    \\
    \,\,\,\,\,\,\,\,\,\,\,\,\,(a) & \,\,\,\,\,\,\,\,\,\,\,\,\,(b) & \,\,\,\,\,\,\,\,\,\,\,\,\,(c) & \,\,\,\,\,\,\,\,\,\,\,\,\,(d) \\
    \end{tabular}
    }
    \caption{Results of tests with
    (a) $d\in\{1,\dots,25\}$, $D=5$, \mbox{$n=1$}, and \mbox{$\epsilon=10^{-14}$}; 
    (b)~\mbox{$d=5$}, $D\in\{1,\dots,25\}$, $n=1$, and $\epsilon=10^{-14}$; 
    (c) \mbox{$d=5$}, $D = 5$, \mbox{$n\in\{1,\dots,25\}$}, and \mbox{$\epsilon=10^{-14}$}; 
    (d) $d=5$, $D=5$, $n=1$, and \mbox{$\epsilon\in\{1,10^{-1},\dots,10^{-25}\}$}.}\label{fig:benchmark}
\end{figure}

\section{Application to prove non-negativity}\label{sec:Application}

One application of our approach to certified evaluations
is to certifiably 
decide whether or not 
all coefficients of the Taylor series expansion 
centered at the origin are non-negative
for a given real rational function $r(x)$. 
We focus on non-negativity since non-positivity is equivalent to non-negativity for $-r(x)$
and alternating in sign is equivalent to non-negativity for $r(-x)$.
The following method uses certified evaluations
to obtain information about the 
coefficients in the tail of the Taylor series expansion
reducing the problem to only needing 
to inspect finitely many coefficients.
This approach assumes
that the function does not have a pole at the origin, 
its denominator has only simple roots,
and its denominator has a real positive root that is strictly
smallest in modulus amongst all its roots.
This approach can be extended to more general settings,
but will not considered~here due to space considerations.

We will make use of the following standard theorem.

\begin{theorem} \label{thm:coefficient-formula}
Let $p(x),q(x)\in\bR[x]$ 
such that $p(x)$ and $q(x)$ have no common root, 
$q(0) \neq 0$ and $\deg(p(x)) < \deg(q(x))=d$. 
If $q(x)$ has only simple roots say $\alpha_1,\dots,\alpha_d\in\bC$,
then $r(x) = p(x)/q(x)$ has a Taylor series 
expansion of the form $r(x) = \sum_{n=0}^\infty r_n x^n$ converging for all $x \in \bC$ with $|x| < \min\{|\alpha_1|, \ldots, |\alpha_d|\}$. Furthermore, for all $n \geq 0$, \begin{equation}\label{eq:coefficient-formula} 
r_n = - \sum_{i=1}^d  \frac{p(\alpha_i)}{\alpha_i q'(\alpha_i)} \alpha_i^{-n}.\end{equation}
\end{theorem}

Theorem~\ref{thm:coefficient-formula} follows
from partial fraction decomposition of rational functions or using linear recurrences.
For completeness, we provide a proof in the Appendix. Using Theorem \ref{thm:coefficient-formula}, we obtain the following result on the eventual behavior of the coefficients of the Taylor series of certain rational functions.

\begin{theorem} \label{thm:eventual-behavior}
With the setup from Theorem~\ref{thm:coefficient-formula},
define $C_i = -p(\alpha_i)/(\alpha_i q'(\alpha_i))$ for 
$i=1,\dots,d$.
If $\alpha_1\in\bR$ is such that 
$|\alpha_1| < \min\{|\alpha_2|,\ldots,|\alpha_d|\}$, 
then there exists~$N$ 
after which exactly one of the following conditions 
on $r_n$ holds:
\begin{enumerate}
    \item If $\alpha_1 >0$ and $C_1 > 0$, then $r_n > 0$ for all $n > N$.
    \item If $\alpha_1 > 0$ and $C_1 < 0$, then $r_n < 0$ for all $n > N$. 
    \item If $\alpha_1 < 0$, then $r_n$ is alternating in sign for all $n > N$, i.e., \mbox{$(-1)^n\cdot r_n > 0$}
    for all $n>N$ or $(-1)^n\cdot r_n < 0$ for all $n > N$. 
\end{enumerate}
Moreover, one may take
$N = \log(K)/\log(M/m)$ 
where 
$K = \sum_{i=2}^d |C_i|/|C_1|$,
$m = |\alpha_1|$, and 
$M = \min\{|\alpha_2|, \ldots, |\alpha_d|\}$.
\end{theorem}

A proof of Theorem~\ref{thm:eventual-behavior}
is provided in the Appendix. Theorems~\ref{thm:coefficient-formula}
and \ref{thm:eventual-behavior} yield the following.

\begin{corollary} \label{cor:eventual-behavior}
Suppose that $f(x),q(x)\in\bR[x]$
have no common root, $q(0)\neq0$, and
$q(x)$ has only simple roots, namely $\alpha_1,\dots,\alpha_d\in\bC$,
such that $\alpha_1\in\bR$ and $|\alpha_1| < \min\{|\alpha_2|, \ldots, |\alpha_d|\}$.  
Let $g(x),p(x)\in\bR[x]$ be the unique polynomials
such that $f(x) = q(x)\cdot g(x) + p(x)$ with
$\deg(p(x)) < \deg(q(x))$.  
Define $C_i = -p(\alpha_i)/(\alpha_i q'(\alpha_i))$ for 
$i =1,\dots,d$.
Then, $f(x)/q(x)$ has a Taylor series expansion 
$f(x)/q(x) = \sum_{n=0}^\infty R_n x^n$ converging for all $x \in \bC$ with $|x| < \min\{|\alpha_1|, \ldots, |\alpha_d|\}$ and there is a threshold $N_0$ so that exactly one of the following conditions on $R_n$ holds:
\begin{enumerate}
    \item If $\alpha_1 >0$ and $C_1 > 0$, then $R_n > 0$ for all $n > N_0$.
    \item If $\alpha_1 > 0$ and $C_1 < 0$, then $R_n < 0$ for all $n > N_0$. 
    \item If $\alpha_1 < 0$, then $R_n$ is alternating in sign for all $n > N_0$, i.e. $(-1)^n\cdot R_n > 0$ for all $n > N_0$ or $(-1)^n\cdot R_n < 0$ for all $n > N_0$. 
\end{enumerate}
One can take
$N_0 = \max\{\deg(f(x))-\deg(q(x))+1,\log(K)/\log(M/m)\}$ 
where 
$K = \sum_{i=2}^d |C_i|/|C_1|$,
$m = |\alpha_1|$, and 
$M = \min\{|\alpha_2|, \ldots, |\alpha_d|\}$.
\end{corollary}
\begin{proof}
Since $f(x)/q(x) = g(x) + p(x)/q(x)$,
applying Theorem~\ref{thm:coefficient-formula} 
yields the first part.
Since the Taylor series coefficients 
of $f(x)/q(x)$ and $p(x)/q(x)$ are same
for $n>\deg(f(x))-\deg(g(x))$, the
second part immediately follows
from Theorem~\ref{thm:eventual-behavior}.
\end{proof}

One key to utilizing Theorem~\ref{thm:eventual-behavior}
and Corollary~\ref{cor:eventual-behavior} is
to certify that $q(x)$ satisfies the requisite
assumptions.
Validating that $q(x)$ has only 
simple roots follows
from computing an irreducible
factorization as in Section~\ref{sec:Roots} 
and checking if every factor has multiplicity $1$.  
Section~\ref{sec:Classification}
describes a certified approach to verify
the remaining conditions on $q(x)$.
Section~\ref{sec:CertNonNeg} yields a complete
algorithm for certifiably deciding non-negativity
of all Taylor series coefficients
which is demonstrated on two examples.

\subsection{Classification of roots}\label{sec:Classification}

Given a polynomial $q(x)\in\bR[x]$ with only simple roots and $q(0)\neq0$, 
the following describes a method to certifiably determine if $q(x)$ has a positive root
that is strictly smallest in modulus amongst all its roots.  
This method uses the ability to certifiably approximate all real points
in zero-dimensional semi-algebraic sets.  Computationally,
this can be accomplished using the command 
{\tt RealRootIsolate} in \maple.

The first step is to certifiably determine if $q(x)$ has a positive root via
$${\cal P} = \{p\in\bR~:~q(p) = 0, p>0\}.$$
If ${\cal P}=\emptyset$, then one returns that $q(x)$ does not have a positive root.
Otherwise, one proceeds to test the modulus condition for $\alpha_1 = \min {\cal P}$.

The modulus condition needs to be tested against negative roots and non-real roots. 
For negative roots, consider
$${\cal N} = \{n\in\bR~:~q(-n) = 0, n>0\}
\,\,\,\,\hbox{and}\,\,\,\,
{\cal B} = \{b\in\bR~:~q(b)=q(-b)=0,b>0\}.$$
By using certified approximations of $\alpha_1$ and points in ${\cal N}$ and ${\cal B}$
of decreasing error, one can certifiably determine which of the following holds:
$\alpha_1 < \min {\cal N}$, $\alpha_1 > \min {\cal N}$, or $\alpha_1\in {\cal B}\subset{\cal P}$.
If $\alpha_1 > \min {\cal N}$ or $\alpha_1\in {\cal B}$, then
one returns that~$q(x)$ does has not a positive root
that is strictly smallest in modulus amongst all its roots.  
Otherwise, one proceeds to the non-real roots by considering
$$
\begin{array}{l}
{\cal L} = \{(r,a,b)\in\bR^3~:~q(r) = 0, r>0, q(a+ib)=0, b>0, a^2+b^2<r^2\} \,\,\,\hbox{and} \\
{\cal E} = \{(r,a,b)\in\bR^3~:~q(r) = 0, r>0, q(a+ib)=0, b>0, a^2+b^2=r^2\}.\end{array}$$
Note that $q(a+ib)=0$ provides two real polynomial conditions on $(a,b)\in\bR^2$ 
via the real and imaginary parts as in \eqref{eq:RealImagQ} so that
${\cal L}$ and ${\cal E}$ are clearly zero-dimensional semi-algebraic sets.
Moreover, for the projection map $\pi_1(r,a,b) = r$, 
$\pi_1({\cal L}\cup{\cal E})\subset {\cal P}$.
By using certified approximations of $\alpha_1$ and points in ${\cal L}$ and ${\cal E}$
of decreasing error, one can certifiably determine 
if $\alpha_1\in \pi_1({\cal L}\cup{\cal E})$ 
or $\alpha_1\notin \pi_1({\cal L}\cup{\cal E})$.
If the former holds, then one returns that~$q(x)$ does has not a positive root
that is strictly smallest in modulus amongst all its roots.  
If the later holds, then one returns that~$q(x)$ does indeed have a positive root
that is strictly smallest in modulus amongst all its roots.  

\subsection{Certification of non-negativity}\label{sec:CertNonNeg}

Suppose that $f(x),q(x)\in\bR[x]$ which satisfy the assumptions
in Corollary~\ref{cor:eventual-behavior}.  The following describes
a method to certifiably determine if all of the 
coefficients $R_n$ of the Taylor series expansion for $f(x)/q(x)$
centered at the origin are non-negative or 
provides an integer $n_0$ such that $R_{n_0}<0$.

First, the Euclidean algorithm is utilized to determine $g(x),p(x)\in\bR[x]$
with $\deg(p(x))<\deg(q(x))$ such that $f(x) = q(x)\cdot g(x)+p(x)$.
Define $h(x)=x\cdot q'(x)$ and $C(x) = -p(x)/h(x)$.
Hence, $d = \deg(q(x)) = \deg(h(x))$ such that $q(x)$ and~$h(x)$ have no common roots.
As in Corollary~\ref{cor:eventual-behavior}, let $\alpha_1,\dots,\alpha_d$
be the roots of $q(x)$ with $\alpha_1\in\bR_{>0}$ such
that $\alpha_1 < \min\{\alpha_2,\dots,\alpha_d\}$.
Let $\beta_1,\dots,\beta_d\in\bC$ (not necessarily all distinct) be the roots of $h(x)$.

Certified evaluations at the roots of $q(x)$ and $h(x)$ with error bound
$\epsilon_k = 2^{-k}$ for $k=1,2,\dots$ can be used 
until the following termination conditions are met:
\begin{enumerate}
    \item\label{S1} $\alpha_i^*$ and $\beta_j^*$ are such that $\alpha_1^*\in\bR$, $|\alpha_i^*-\alpha_i|<\epsilon_k$, and $|\beta_j^*-\beta_j|<\epsilon_k$
    \item\label{S2} the set $\{0, \alpha_1^*, \ldots, \alpha_d^*\}$ is $2\cdot\epsilon_k$ separated, i.e., 
    $|s-t|^2 \geq (2 \epsilon_k)^2$ for all distinct $s,t$ in this set,
    \item\label{S3} $\gamma^* \leq \min\{|\alpha_i^* - \beta_j^*|~:~1 \leq i,j \leq d\}$
    such that $\gamma^*>2\cdot\epsilon_k + \epsilon_k^{1/(4d)}$,
    \item\label{S4} for $m^* = \alpha_1^* + \epsilon_k$ and $M^* \leq \min\{|\alpha_2^*|,\dots,|\alpha_d^*|\}-\epsilon_k$,
    one has $m^* < M^*$,
    \item\label{S5} $C_i^*$ such that $|C_i^*-C(\alpha_i)|<\epsilon_k$,
    \item\label{S6} $L_i^*$ such that $L_i^*\geq |c_d|^{-2}\sum_{\ell=0}^{2d-1}|u^{(\ell)}(\alpha_i^*)|\epsilon_k^\ell/\ell!$
    where $c_d$ is the leading coefficient of $q(x)$ and $u(x) = -p'(x)h(x)+p(x)h'(x)$, and
    \item\label{S7} either (a) $C_1^*+L_1^* \sqrt{\epsilon_k} < 0$ or (b) $C_1^*+L_1^*\sqrt{\epsilon_k} > 0$.
\end{enumerate}
Before proving that such a termination condition can be met,
we describe the last steps which are justified
by Corollary~\ref{cor:eventual-behavior}.  Let $\epsilon=\epsilon_k$ be the value where the termination conditions
are met and calculate $0 < b^* \leq \min\{|C_1^* \pm L_1^* \sqrt{\epsilon}|\}$,
\[
\begin{array}{c}
K^* \geq (1/b^*) \sum_{i=2}^d (|C_i^*| + L_i^* \sqrt{\epsilon}), \,\,\, A^* \geq \log(K^*)/\log(M^*/m^*),\\[0.1in]
N_0^* = \max\{\deg(f(x))-\deg(q(x)), \lceil A^* \rceil\}.
\end{array}
\] 
The inequalities in Items~\ref{S3},~\ref{S4}, and~\ref{S6} 
and these values above are meant to signify 
the rounding direction in machine precision used to compute the values.
With this, if $C_1^*+L_1^* \sqrt{\epsilon_k} < 0$, then return $n_0 = N_0^*+1$ 
in which $R_{n_0} < 0$.
Otherwise, the non-negativity of all $R_n$ is equivalent to the non-negativity of
$R_0,\dots,R_{N_0^*}$ which can be computed by explicit computation.
If there exists $n_0\in\{0,\dots,N_0^*\}$ such that $R_{n_0}<0$, return $n_0$.
Otherwise, return that all $R_n$ are non-negative.

\medskip

Returning to the proof of correctness, note that Items~\ref{S1},~\ref{S5}, and~\ref{S6} follow
from evaluation errors and 
$$\delta = \min\{|\alpha_i-\alpha_j|, |\alpha_i|~:~ 1\leq i < j \leq d\} > 0.$$  
By Item~\ref{S1}, $|\alpha_i^* - \alpha_j^*| \geq |\alpha_i - \alpha_j| - 2 \epsilon_k$ 
and $|\alpha_i^*| \geq |\alpha_i| - \epsilon_k$.  Thus, 
$$\delta^* = \min\{|\alpha_i^*-\alpha_j^*|, |\alpha_i^*|~:~1\leq i < j \leq d\} \geq \delta - 2 \epsilon_k - \eta$$ 
where $\eta > 0$ is the machine precision on the lower bounds on the quantities in $\delta^*$. 
Since $\epsilon_k \to 0$ and $\eta$ can be made arbitrarily small, eventually $\delta^* - 2 \epsilon_k -\eta > 2 \epsilon_k$ and Item~\ref{S2} will be met.

Since $q(x)$ and $h(x)$ have no roots in common, consider 
$$\gamma = \min\{|\alpha_i - \beta_j|~:~1 \leq i,j \leq d\} > 0.$$ 
We have $\gamma^* \geq \gamma - 2 \epsilon_k - \eta$ where $\eta > 0$ is the machine precision on the lower bounds of the quantities in $\gamma^*$. Eventually, $\gamma^* \geq \gamma - 2 \epsilon_k - \eta > 2 \epsilon_k + \nu + \epsilon_k^{1/(4d)}$ where $\nu > 0$ is the machine precision on the upper bound for~$\epsilon_k^{1/(4d)}$ and Item~\ref{S3} will be met.

Since $\Delta = M - m > 0$,
we have $M^* > M - \epsilon_k - \eta$ where $\eta > 0$ is the machine precision on the lower estimates in $M^*$ and also $m^* = \alpha_1^* + \epsilon_k < \alpha_1 + 2 \epsilon_k$.  
Thus, $M^* - m^* \geq \Delta - 3 \epsilon_k - \eta$ so that eventually Item~\ref{S4} will be met.

Since $C(x) = p(x)/h(x)$, we have $C'(x) = u(x)/h^2(x)$.
Assuming the previous items have all been met, we have
$$|\beta_j - \alpha_i^*| \geq |\beta_j^* - \alpha_i^*| - |\beta_j - \beta_J|^* > 2 \epsilon_k + \epsilon_k^{1/(4d)} - \epsilon_k > \epsilon_k.$$  
Hence, $h(x)$ has no roots in $B(\alpha_i^*, \epsilon_k)$ and $C'(x)$ 
is continuous on $B(\alpha_i^*,\epsilon_k)$.
Fix $z\in B(\alpha_i^*,\epsilon_k)$ and 
let $\zeta$ be the straight line segment contour from 
$\alpha_i^*$ to $z$ in $B(\alpha_i^*,\epsilon_k)$.  
Since $C'(x)$ exists on $B(\alpha_i^*,\epsilon_k)$, 
$C(z) - C(\alpha_i^*) = \int_{\zeta} C'(x) dx$.  
Thus, $|C(z) - C(\alpha_i^*)| \leq P\cdot |z - \alpha_i^*|$ where $P = \max \{ |C'(x)|~:~x \in B(\alpha_i^*,\epsilon_k)\}$. 
Therefore, we have $P \leq P_1/P^2_2$ where $P_1 = \max \{|u(x)|~:~x \in B(\alpha_i^*,\epsilon_k)\}$ and $P_2 = \min \{|h(x)|~:~x \in B(\alpha_i^*,\epsilon_k)\}$.  
Since $u(z) = \sum_{\ell=0}^{2d-1} u^{(\ell)}(\alpha_i^*) (z - \alpha_i^*)^\ell/\ell!$, 
we have $P_1 \leq \sum_{\ell=0}^{2d-1} |g^{(\ell)}(\alpha_i^*)| \epsilon_k^\ell/\ell! ( 1 + \eta)$ where $\eta > 0$ is the machine precision that results from the upper bound on the quantities $|u^{(\ell)}(\alpha_i^*)|$.  
Since 
$$h(z) = c_d \prod_{j=1}^d(z - \beta_j)\,\,\,\hbox{and}\,\,\,
|z - \beta_j| \geq |\alpha_i^* - \beta_j^*| - |z-\alpha_i^*| - |\beta_j^* -\beta_j| > \epsilon_k^{1/(4d)},$$
we have $|h(z)| \geq |c_d| \epsilon^{1/4}$.  Thus, $P^2_2 \geq |c_d|^2 \epsilon^{1/2}$ and $|C(z) - C(\alpha_i^*)| \leq L_i^* \sqrt{\epsilon}$. Since $B(\alpha_i^*,\epsilon_k) \subset B(\alpha_i,2 \epsilon_k) \subset B(\alpha_i,1)$, all estimates $|u^{(\ell)}(\alpha_i^*)|$ will be bounded above by the corresponding maximum values of $|u^{(\ell)}(z)|$ for $z\in B(\alpha_1,1)$.  Thus, even though $P_1$ varies in each step $k$, $P_1$ and $L_1^*$ will be uniformly bounded above for all $k$.  By the proof of Theorem \ref{thm:eventual-behavior}, $C_1 \neq 0$ so Item~\ref{S7}(a) will eventually be met if $C_1 < 0$
and Item~\ref{S7}(b) will eventually be met if $C_1 > 0$.

\medskip

After having proving termination, we note 
that in order to get a value of $N_0^*$ which is reasonable close to the
value of $N_0$ in Corollary~\ref{cor:eventual-behavior},
one may continue to decrease $\epsilon_k$ past the point where all termination conditions are first met.
The reason for this is to separate $m^*$ and $M^*$ as far all possible, i.e., to match the actual gap between $m$ and $M$ as closely as possible. 
This could be wise especially when $m^*$ is very close to $M^*$
in which case $\log(M^*/m^*)$ will be very close to $0$ so that $N_0^*$ will be very large.

\begin{example}\label{ex:NonNegEx}

To demonstrate the approach, consider the rational functions
$$(1-x^3-x^7+x^{18})^{-1} \,\,\,\,\,\,\,\,\,\,\,\,\,\,\,\,\,\,
\hbox{and} \,\,\,\,\,\,\,\,\,\,\,\,\,\,\,\,\,\,
(1-x^3-x^7+x^{21})^{-1}.
$$
The implementation of this approach in \maple \  
certifies that both rational functions 
have Taylor series expansions 
centered at the origin
where all of the coefficients are non-negative.
The value of $N_0$ 
from Corollary~\ref{cor:eventual-behavior}
which could be certified 
by the method described above was $N_0^*=204$ and $N_0^*=55$, respectively.
Thus, it was easy to utilize {\tt series} in \maple \
to check the non-negativity of the Taylor series 
coefficients up to~$N_0^*$ combined with Corollary~\ref{cor:eventual-behavior}
for the~tail.
\end{example}

\section{Conclusion}\label{sec:Conclusion}

This manuscript developed techniques for certified evaluations 
of locally H\"older continuous functions
at roots of polynomials along with an implementation in {\tt Maple}.  
These techniques were demonstrated on several problems
including certified bounds on critical values
and proving non-negativity of coefficients in
Taylor series expansions.  
Although this paper focused on roots of univariate polynomials, it is natural to extend 
to multivariate polynomial systems in~the~future.  

\section*{Acknowledgments}

JDH was supported in part by
NSF CCF 1812746. CDS was supported in part by Simons Foundation grant 360486.

\bibliographystyle{splncs04}

\section*{Appendix}

{\noindent \em Proof of Theorem~\ref{thm:coefficient-formula}.} 
Suppose that $C \neq 0$ such that $q(x) = C\cdot \prod_{i=1}^d (x - \alpha_i)$.  
Thus, we know 
\mbox{$q'(x) = C\cdot \sum_{i=1}^d \prod_{j \neq i} (x - \alpha_j)$} and $q'(\alpha_i) = C\cdot \prod_{j \neq i} (\alpha_i - \alpha_j) \neq 0$ for all $i$.
Let $p_i(x) = q(x)/(x - \alpha_i) = C\cdot\prod_{j \neq i} (x - \alpha_j)$.  
Hence, $p_i(\alpha_i) = q'(\alpha_i)$ and $p_i(\alpha_j) = 0$ if $j \neq i$. 
The polynomials $p_1,\dots,p_d$ are linearly independent 
since, if $\sum_{i=1}^d a_i p_i(x) = 0$, then evaluating at $x = \alpha_j$
yields $a_j \cdot q'(\alpha_j) = 0$ which implies $a_j = 0$.  
Thus, they must form a basis for the $d$-dimensional vector space of polynomials of degree at most 
$d-1$.

Since $p(x)$ has degree at most $d-1$, there are unique constants $a_i$ 
so that $\sum_{i=1}^d a_i p_i(x) = p(x)$.  
Evaluating at $x = \alpha_j$ yields $a_j q'(\alpha_j) = p(\alpha_j)$ so that
$a_j = p(\alpha_j)/q'(\alpha_j)$.  Therefore, for all $x\in\bC\setminus\{\alpha_1,\dots,\alpha_d\}$,
\begin{equation}\label{eq:partial-fraction-expansion} 
\hbox{\footnotesize $\displaystyle
\frac{p(x)}{q(x)} = \sum_{i=1}^d \frac{p(\alpha_i)}{q'(\alpha_i)} \frac{1}{x-\alpha_i} = \sum_{i=1}^d - \frac{p(\alpha_i)}{\alpha_i q'(\alpha_i)} \frac{1}{1 -x/ \alpha_i}.$}\end{equation} 

The terms in \eqref{eq:partial-fraction-expansion} have a Taylor series 
expansion centered at the origin that converge for all $x$ with $|x| < \min\{|\alpha_1|, \ldots, |\alpha_d|\}$ such that, as \eqref{eq:coefficient-formula} claims,
\[\hbox{\footnotesize $\displaystyle
\frac{p(x)}{q(x)} = \sum_{i=1}^d - \frac{p(\alpha_i)}{\alpha_i q'(\alpha_i)} \sum_{n=0}^\infty \alpha_i^{-n} x^n = \sum_{n=0}^\infty \left(-\sum_{i=1}^d  \frac{p(\alpha_i)}{\alpha_i q'(\alpha_i)} \alpha_i^{-n} \right) x^n.$} \]

\medskip\medskip

{\noindent \em Proof of Theorem~\ref{thm:eventual-behavior}.}
Clearly, one has $r_n = \frac{d^n}{dz^n}(p(z)/q(z))|_{z=0}$.  Since $p(x)$ and~$q(x)$ have real coefficients, $r_n$ is real for all $n \geq 0$.  
For $i\in\{1,\dots,d\}$, let $t^i_n = C_i \alpha_i^{-n}$ so that~\eqref{eq:coefficient-formula} 
reduces to $r_n = \sum_{i=1}^d t^i_n$. 
Moreover, $\alpha_1\in\bR\setminus\{0\}$ implies $C_1\in\bR\setminus\{0\}$.
Clearly, if $\alpha_1<0$, then $t^1_n$ is alternating in sign.

Consider the case when $\alpha_1>0$.
First, note that $t^1_n$ and $C_1$ always have the same sign.
The following derives a threshold $N$ such that $|r_n - t^1_n| < |t^1_n|$ for all $n > N$.  
Given such an $N$, $r_n$ will have the same sign as $t^1_n$ and $C_1$ 
for $n > N$ and the theorem will be proved.
To that end, since $(r_n - t^1_n)/t^1_n = \sum_{i=2}^d t^i_n/t^1_n$, 
\[\hbox{\footnotesize $\displaystyle
\frac{|r_n - t^1_n|}{|t^1_n|} \leq \sum_{i=2}^d \frac{|C_i|}{|C_1|}\frac{|\alpha_1|^n}{|\alpha_i|^n} \leq K \left(\frac{m}{M}\right)^n$}\] for all $n$. Since, by assumption, $m/M < 1$, there is a threshold $N$ so that $K (m/M)^n < 1$ and $|r_n - t^1_n| < |t^1_n|$ for all $n > N$.  We may take $N$ so that $K(m/M)^N=1$ or $N = \log(K)/\log(M/m)$ as claimed.

\end{document}